\documentclass[runningheads,envcountsame,a4paper]{llncs}

\usepackage[dvipsnames]{xcolor}
\usepackage{graphicx}

\usepackage{amsmath}
\usepackage{amsfonts}
\usepackage{mathtools}
\usepackage{stmaryrd}

\usepackage{mathpartir}
\usepackage{bussproofs}

\usepackage{caption}
\usepackage{subcaption}
\usepackage{xcolor}

\usepackage{paralist}
\usepackage{xspace}
\usepackage{cleveref}

\usepackage{booktabs}
\usepackage{longtable,tabu}

\usepackage{tikz}
\usetikzlibrary{circuits.logic.US,shapes.gates.logic.US,calc,trees,positioning,arrows,matrix}

\usepackage{url}

\usepackage[acronyms,shortcuts,hyperfirst=false]{glossaries}

\newacronym{dp}    {DP}    {defect prediction}
\newacronym{fd}    {FD}    {fault distribution}
\newacronym{ft}    {FT}    {fault tree}
\newacronym{loc}    {LOC}    {lines of code}
\newacronym{sp}    {SP}    {starting point}
\newacronym{srgm}    {SRGM}    {software reliability growth model}
\newacronym{tra}    {TRA}    {test resource allocation}
\newacronym{trap}    {TRAP}    {test resource allocation problem}
\newacronym{qa}    {QA}    {quality assurance}
\newacronym{qcl}    {QCL}    {Quantitative Confidence Logic}
\newacronym{cps}    {CPS}    {cyber-physical system}

\newcommand{\Astrahl}{\texttt{Astrahl}\xspace}

\newcommand{\eg}{e.\,g.\xspace}
\newcommand{\ie}{i.\,e.\xspace}

\renewcommand{\phi}{\varphi}
\renewcommand{\epsilon}{\varepsilon}

\DeclareMathOperator{\Meas}{Meas}
\newcommand{\sep}{\,|\,}
\newcommand{\ens}[1]{\left\{{#1}\right\}}
\newcommand{\enscomp}[2]{\left\{{#1}\,\middle|\,{#2}\right\}}
\newcommand{\R}{\mathbb{R}}

\newcommand{\B}{\mathbb{B}}
\newcommand{\C}{\mathbb{C}}
\newcommand{\Prop}{\mathrm{Prop}}
\newcommand{\Form}{\mathrm{Form}}
\newcommand{\logic}{\ac{qcl}\xspace}
\newcommand{\software}{\textnormal{software}}
\newcommand{\hardware}{\textnormal{hardware}}
\newlength{\botlength}
\newcommand{\undet}{\bot\hspace{-\botlength}\top}
\newcommand{\imply}{\Rightarrow}
\newcommand{\smsp}{\Omega}
\newcommand{\sfld}{\mathfrak{F}}
\newcommand{\meas}{P}
\newcommand{\mbool}{2}
\newcommand{\prob}[1]{\meas[#1]}
\newcommand{\probin}[2]{\prob{#1\in #2}}
\newcommand{\probeq}[2]{\prob{#1=#2}}
\newcommand{\semconf}[4]{#3\colon #4}

\newcommand{\axiom}{\textnormal{(\textit{ax})}}

\newcommand{\impi}{\textnormal{($\imply_\textnormal{I}$)}}
\newcommand{\impel}{\textnormal{($\imply_\textnormal{E,\textit{l}}$)}}
\newcommand{\imper}{\textnormal{($\imply_\textnormal{E,\textit{r}}$)}}
\newcommand{\andi}{\textnormal{($\wedge_\textnormal{I}$)}}

\newcommand{\ori}{\textnormal{($\vee_\textnormal{I}$)}}

\newcommand{\unk}{\textnormal{(\textit{unk})}}
\newcommand{\topr}{\textnormal{($\top_\textnormal{I}$)}}
\newcommand{\botr}{\textnormal{($\bot_\textnormal{I}$)}}
\newcommand{\negi}{\textnormal{($\neg_\textnormal{I}$)}}

\newcommand{\sem}[2]{\ensuremath{\llbracket{#1}\rrbracket_{#2}}}

\makeatletter
\renewcommand\paragraph{\@startsection{paragraph}{4}{\z@}%
    {-12\p@ \@plus -4\p@ \@minus -4\p@}%
    {-0.5em \@plus -0.22em \@minus -0.1em}%
    {\normalfont\normalsize\bfseries}}
\makeatother

\begin{document}

\title{Architecture-Guided Test Resource Allocation Via Logic
\thanks{The authors are supported by ERATO HASUO Metamathematics for
Systems Design Project (No. JPMJER1603).}}

\author{Clovis Eberhart\inst{1,2}
\and
Akihisa Yamada\inst{3}
\and
Stefan Klikovits\inst{1}
\and
Shin-ya Katsumata\inst{1}
\and
Tsutomu Kobayashi\inst{4,1}
\and
Ichiro Hasuo\inst{1}
\and
Fuyuki Ishikawa\inst{1}
}

\authorrunning{C. Eberhart et al.}

\institute{National Institute of Informatics, 2-1-2 Hitotsubashi,
Chiyoda-ku, Tokyo 101-8430, Japan\\
\email{\{eberhart,klikovits,s-katsumata,t-kobayashi,hasuo,f-ishikawa\}@nii.ac.jp}
\and
Japanese-French Laboratory for Informatics, Tokyo, Japan
\and
National Institute of Advanced Industrial Science and Technology,
2-3-26, Aomi, Koto-ku, Tokyo 135-0064, Japan\\
\email{akihisa.yamada@aist.go.jp}
\and
Japan Science and Technology Agency,
4-1-8, Honcho, Kawaguchi-shi, Saitama 332-0012, Japan
}

\maketitle
\setcounter{footnote}{0}

\begin{abstract}
  We introduce a new logic named Quantitative Confidence Logic (QCL)
  that quantifies the level of confidence one has in the conclusion
  of a proof.
  By translating a fault tree representing a system's architecture
  to a proof, we show how to use QCL to give a solution to the
  \acl{trap} that takes the given architecture into account.
  We implemented a tool called \Astrahl{} and compared our results to
  other testing resource allocation strategies.

  \keywords{Reliability \and
  Test Resources Allocation \and
  Logic.}
\end{abstract}

\section{Introduction}
\label{sec:intro}
With modern systems growing in size and complexity, asserting their
correctness has become a paramount task, and
despite advances in the area of formal verification, testing remains a
vital part of the system life cycle due to its versatility,
practicality, and low entry barrier.
Nevertheless,
as test resources are limited, it is an important task to most
effectively allocate them among system components, a problem commonly
known as the \ac{trap} (see \eg~\cite{pietrantuono2020testing}).
In this paper we formulate the \ac{trap} as follows: given a system
that consists of multiple components and a certain, limited amount
of test resources (\eg time or money), how much of the budget should
we allocate to each component in order to minimise the chance of
system failure, \ie to increase its reliability.

We propose an approach to the \ac{trap}, based on a novel logic system
called \logic.
\logic differs from classical logic, in that a proof tree does not
conclude truth from assumptions, but rather analyses how
\emph{confidence} is propagated from assumptions to conclusions.
We prove key soundness properties of \ac{qcl} with respect to a
probabilistic interpretation (Section~\ref{sec:qcl}).

Learning a new logic is a hard task, especially to practitioners.
Thus we do not demand users to learn \ac{qcl},
but take it as an intermediate language to which already accepted representations of system architectures are translated.
As an example, we show how to translate the well-known concept
of \acp{ft} into \ac{qcl} proof trees (Section~\ref{sec:trap}).

We then formulate the \ac{trap} as an optimisation problem with
respect to a given \ac{ft}, translated to a \ac{qcl} proof tree
(Section~\ref{sec:optim}).
Here we allow users to specify a \emph{confidence function} for each
component, which describes how an amount of spent test resources
relates to an increase in that component's reliability.

We implement our approach as a tool (\Astrahl) that takes as input an
\ac{ft}, a confidence function for each component, the current
confidence in each component, and the total amount of test resources
the user plans to spend.
Then \Astrahl outputs a proposed allocation of the test resources over
components.
We validate our approach through experiments (\Cref{sec:exp}).

An advantage of our method is that it is not tied to a fixed confidence
function, and can therefore assign different confidence functions to
different components.
This will be useful in modelling systems with highly heterogeneous
components such as \acp{cps}; for example, hardware components would demand more effort to increase confidence than software, and would depend on the type of components or their vendors.

We expect \Astrahl to be used continuously in a system's development:
the confidence in each component increases as they pass more tests.
By rerunning \Astrahl with updated component confidences, we obtain a
test resource allocation strategy.
We expect our approach to be promising in product line development,
where a number of system configurations are simultaneously developed
with common components.
In such a situation, updating confidence in a component for one system
positively impacts the test strategies for other systems.

\subsection{Related Work}

\subsubsection{Test Resource Allocation}

Most approaches to the \ac{trap} use \acp{srgm} such as models based
on Poisson Process (\eg~\cite{goel1979time}) to capture
the relationship between testing efforts and reliability growth,
or in our words, confidence functions.
A typical \ac{trap} approach formulates the problem using a particular
\ac{srgm} and provides a solution using exact
optimisation~\cite{huang2006optimal} or a metaheuristic such as a
genetic algorithm~\cite{zhang2017constraint}.
A challenge in this area is to take the structure (inter-module
relations) of the target system into account; existing
studies consider particular structures such as parallel-series
architecture~\cite{zhang2017constraint} or Markovian
architecture~\cite{pietrantuono2010software}.
In addition, there is high demand for dynamic allocation methods
(\eg \cite{carrozza2014dynamic})
because in practice \acp{srgm}, system structures, and testing
processes often become different from those planned at first.
Our approach has multiple beneficial features over the existing approaches: (1) it is
independent of particular \acp{srgm} and optimisation strategies, (2)
it can take complex structure into account using \acp{ft}, and (3) it
can be used for dynamic allocation.

\subsubsection{Fuzzy Logics}

Fuzzy logics~\cite{hajek2013metamathematics} is a branch of logic
interested in deductions where Boolean values are too coarse.
In its standard semantics~\cite{esteva2001monoidal}, formulas are
given a numeric truth value in the interval $[0,1]$, where $0$ represents falsity
and $1$ truth.
These numerical values can be used to represent the confidence one has
in an assertion: we give high values to propositions we are confident
are true, and low values to those we are confident are false.
This is slightly different from our approach, where $1$ corresponds to
confidence (either in truth or falsity) and $0$ to absence of
knowledge.

\subsubsection{Dempster-Shafer Theory}

Dempster-Shafer theory~\cite{dempster1967upper,shafer1976mathematical}
is a mathematical theory of belief.
One of its characteristic features is that if one has a belief $b$ in
an assertion, they can have any belief $b' \leq 1-b$ in its negation,
contrary to traditional Bayesian models, where it is necessarily
$1-b$.
This feature is crucial to model uncertainty due to absence of
knowledge, and Dempster-Shafer theory has been used to model
reliability in engineering contexts~\cite{sallak2013reliability}.
Our approach draws inspiration from fuzzy and three-valued logics to
model this feature.

\subsubsection{Fault Tree Analysis}

Fault trees~\cite{vesely1981fault} are tree structures that represent
how faults propagate through a system.
In qualitative \ac{ft} analysis, they are used to determine root
causes~\cite{ericson1999fault}.
In quantitative \ac{ft} analysis, basic events are assigned fault
probabilities, and the overall system failure probability is given by
propagating the fault probabilities through the fault tree.
Our approach uses the same ingredients (assigning numeric values to
basic events and propagating them through the fault tree), but
repurposed to solve another problem.

\section{Quantitative Confidence Logic}
\label{sec:qcl}
\newcommand{\snote}[1]{\textcolor{blue}{#1}}

This section introduces \logic, which we use throughout this paper.
A \logic formula is a standard propositional formula $\phi$
equipped with a pair of reals,
written $\phi \colon (t,f)$,
where $t \in [0,1]$ represents our confidence that
$\phi$ holds and $f \in [0,1]$ the confidence that $\phi$ does not
hold, so $t+f$ represents how much confidence one has about $\phi$,
and $1-t-f$ lack of confidence about $\phi$.
Absolute confidence is represented by $1$ and total absence of
knowledge by $0$, so that $\phi \colon (1,0)$ means full trust
that $\phi$ holds, $\phi \colon (0,0)$ means we have no knowledge
about $\phi$, and $\phi \colon (1/2,1/2)$ represents the fact that we
know with very high confidence that $\phi$ holds with 50\% chance.

In Section~\ref{subsec:qcl:proofs} we define the syntax of \logic and introduce its proof rules.
We also show how to derive standard proof rules from them.
In Section~\ref{subsec:qcl:prob}
we give a probabilistic interpretation of \logic formulas and show
that \logic proof rules are sound with respect to it.
We also explain how the particular shapes of the rules serve as the
basis of our optimisation algorithm (\Cref{sec:optim}).

\subsection{Syntax and proof rules of QCL}\label{subsec:qcl:proofs}

This section introduces \ac{qcl}, starting with \emph{formulas with
confidences}, then sequents, and finally proof rules.

\begin{definition}
  Given an arbitrary set of atomic propositions $\Prop$, formulas are
  defined inductively by the following grammar:
  \begin{equation*}
    \phi \Coloneqq
    A \sep
    \top \sep
    \bot \sep
    \phi \imply \phi\rlap{,}
  \end{equation*}
  for $A \in \Prop$.
  We denote by $\Form$ the set of all formulas.
\end{definition}

As in classical logic, negation, disjunction, and conjunction can be defined by
syntactic sugar $\neg \phi \equiv \phi \imply \bot$, $\phi \lor \psi
\equiv \neg \phi \imply \psi$, and $\phi \land \psi \equiv \neg (\neg
\phi \lor \neg \psi)$.

We now equip such formulas with confidence.
We define the \emph{space of confidences} as $\C = \enscomp{(t,f) \in
[0,1]^2}{t+f \leq 1}$.
\begin{definition}
  A \emph{formula with confidence} is a pair $(\phi,c) \in \Form
  \times \C$, written $\phi \colon c$.
  For $\phi \colon (t,f)$, we call $t$ the \emph{true} confidence
  in $\phi$ and $f$ its \emph{false} confidence.
\end{definition}
Intuitively, a formula with confidence $\phi \colon (t,f)$ represents
the fact that our confidence that $\phi$ holds is $t$, and our
confidence that $\phi$ does not hold is $f$.
This should mean that the chance that $\phi$ holds is at least $t$,
and the chance that it does not is at least $f$.
Another way to look at confidences is \emph{intervals of probability}.
Each confidence $(t,f)$ bijectively determines a sub-interval
$[t,1-f]$ of $[0,1]$.
Then a formula with confidence $\phi:(t,f)$ represents that the
probability of $\phi$ being true is within the interval $[t,1-f]$.
We make this intuition more concrete in Section~\ref{subsec:qcl:prob},
where we give an interpretation of \logic in terms of probabilities.

Equipping formulas with numeric values is reminiscent of fuzzy
logics\footnote{Here, we mean fuzzy logics interpreted in $[0,1]$.}.
To explain the fundamental difference between our approach and fuzzy
logics, let us consider two orders on $\C$.
Let $(t,f) \sqsubseteq (t',f')$ if $t \leq t'$ and $f \leq f'$; we
call $\sqsubseteq$ the \emph{confidence order}, as $c \sqsubseteq c'$
holds exactly when $c'$ represents more confidence (both true and
false) than $c$.
Our approach is centred around $\sqsubseteq$, since we are interested
in ``how confident'' we are in an assertion.
Similarly, let $(t,f) \leq (t',f')$ if $t \leq t'$ and $f \geq f'$; we
call $\leq$ the \emph{truth order}, as $c \leq c'$ intuitively means
that $c'$ is ``more true'' than $c$.
Fuzzy logics is centred around the truth order $\leq$ (especially on
elements of the form $(t,1-t)$), as it is a logic about how true
assertions are.

One way to link our approach to fuzzy logics is via
three-valued logics~\cite{bergmann2008introduction}.
Fuzzy logics can be seen as equipping formulas with a numeric truth
value $t \in [0,1]$ and a falsity value $f \in [0,1]$ such that $t+f =
1$.
This is equivalent to equipping formulas only with a numeric value $t
\in [0,1]$, while $f$ is the implicit difference to $1$.
With three-valued logics, formulas have three possible outcomes:
truth $\top$, falsity $\bot$, and uncertainty $\undet$, and each is
given a value $t$, $f$, and $u$, such that $t+f+u = 1$, which is
equivalent to equipping them with $(t,f) \in [0,1]^2$ such that $t+f
\leq 1$, while $u$ is implicit.
Dempster-Shafer theory is similar, with $t$ representing our belief in
$\phi$, $f$ our belief in $\neg \phi$, and $u$ our degree of
uncertainty.

We now introduce the sequents on which \logic operates.
\begin{definition}
  A \emph{sequent} in \logic, written $\Gamma \vdash \phi \colon c$,
  consists of a finite set $\Gamma \subseteq \Form\times\C$ of
  formulas with confidences (written as a list), a formula $\phi \in
  \Form$, and a confidence $c\in\C$.
\end{definition}
Such a sequent intuitively means that, if all formulas with
corresponding confidences in $\Gamma$ hold,
then $\phi$ holds with confidence $c$ as well.

\begin{definition}
\begin{figure}
  \framebox[\textwidth]{
  \begin{mathpar}
    \inferrule*[right=\axiom]{ }{\Gamma, \phi \colon (t,f) \vdash
      \phi \colon (t,f)}

    \inferrule*[right=\unk]{ }{\Gamma \vdash \phi \colon (0,0)}

    \inferrule*[right=\topr]{ }{\Gamma \vdash \top \colon (1,0)}

    \inferrule*[right=\botr]{ }{\Gamma \vdash \bot \colon (0,1)}

    \inferrule*[right=\impi]{\Gamma \vdash \phi \colon (t,f) \\
      \Gamma \vdash \psi \colon (t',f')}{\Gamma \vdash \phi \imply
      \psi \colon (f + t' - f t',t f')}

    \inferrule*[right=\impel{\rm\ if $t' \neq 0$ and $f' \neq 1$}]{\Gamma \vdash \phi \imply \psi \colon (t,
      f) \\ \Gamma \vdash \phi \colon (t',f')}{\Gamma \vdash \psi
      \colon \left(1-\frac{1-t}{t'},\frac{f}{1-f'}\right)}

    \inferrule*[right=\imper{\rm\ if $t' \neq 1$ and $f' \neq 0$}]{\Gamma \vdash \phi \imply \psi \colon (t,
      f) \\ \Gamma \vdash \psi \colon (t',f')}{\Gamma \vdash \phi
      \colon \left(\frac{f}{1-t'},1-\frac{1-t}{f'}\right)}
  \end{mathpar}
  }
  \caption{Proof rules of \acl{qcl}}
  \label{fig:rules}
\end{figure}

\emph{Proof trees} in \logic are built from the \logic proofs rules given in
\Cref{fig:rules}.
There, the notation $\chi \colon (t'',f'')$ in conclusions is a
shorthand for
$$\chi \colon (\min(\max(t'',0),1), \min(\max(f'',0),1))\rlap{.}$$
\end{definition}

Note that $(t'',f'') \in \C$ because $t'' + f'' \leq 1$ in all rules.
Note also that rules $\impel$ and $\imper$ are conditioned so that
confidence values do not contain indeterminate forms $0 / 0$.
These side conditions correspond to the facts that, if $\phi \imply \psi$ is true,
$\phi$ being false gives no information about $\psi$, and $\psi$ being
true gives no information about $\phi$.

$\axiom$, $\topr$, and $\botr$ are self-explanatory, while
$\unk$ states that anything can be proved, but with null confidence.
One way to think about $\impi$ is that, if $\phi$ and $\psi$ are
independent (in a way made precise in Section~\ref{subsec:qcl:prob}),
$\phi$ holds with probability in $[t, 1-f]$ and $\psi$ with
probability in $[t', 1-f']$, then $\phi \imply \psi$ holds with
probability in $[f + t' - ft', t f']$.
The elimination rules are designed similarly.

\begin{figure}
  \framebox[\textwidth]{
  \begin{mathpar}
    \inferrule*[right=\negi]{\Gamma \vdash \phi \colon (t,f)}{\Gamma
      \vdash \neg \phi \colon (f,t)}

    \inferrule*[right=\andi]{\Gamma \vdash \phi \colon (t,f) \and
      \Gamma \vdash \psi \colon (t',f')}{\Gamma \vdash \phi \land
      \psi \colon (t t', f + f' - f f')}

    \inferrule*[right=\ori]{\Gamma \vdash \phi \colon (t,f) \and
      \Gamma \vdash \psi \colon (t',f')}{\Gamma \vdash \phi \lor \psi
      \colon (t + t' - t t', f f')}
  \end{mathpar}
  }
  \caption{Derivable introduction rules}
  \label{fig:more_rules}
\end{figure}
From the rules in \Cref{fig:rules} and the encodings of negation,
disjunction, and conjunction, we can derive the introduction rules in
\Cref{fig:more_rules} (we can also derive elimination rules, but do
not discuss them here).
A point worth attention is that the shape of these rules is quite
unorthodox.
In particular, one should not need to have confidence in both
disjuncts to have confidence in a disjunction.
However, this unorthodox shape is exactly the reason why our solution
to the \ac{trap} (defined in Section~\ref{sec:optim}) works well, as
we show in \Cref{ex:trap}.
Moreover, we can derive a disjunction from a single disjunct, as:
\begin{center}
  \AxiomC{$\Gamma \vdash \phi \colon (t,f)$}
  \AxiomC{}
  \RightLabel{\unk}
  \UnaryInfC{$\Gamma \vdash \psi \colon (0,0)$}
  \RightLabel{\ori,}
  \BinaryInfC{$\Gamma \vdash \phi \lor \psi \colon (t + 0 - t \cdot
    0,f \cdot 0) = (t,0)$}
  \DisplayProof
\end{center}
which represents (one of) the usual disjunction introduction
rules.

\begin{remark}
  \ac{qcl} rules are different from those of classical logic and do
  not extend them.
  Because the rules' design is strongly centred around how confidence
  flows from hypotheses to conclusions and based on independence of
  hypotheses, \ac{qcl} cannot prove sequents such as $\emptyset \vdash
  A \imply A \colon (1,0)$.
  How to allow reasoning about both confidence and truth in the same
  logic is left for future work.
\end{remark}

\begin{example}
  \label{ex:cps}
  The correctness of a system with software and hardware components is
  based on the correctness of
  both components.
  A classical logical proof that the system is correct assuming both
  its software and hardware are correct is:
  \begin{center}
    \AxiomC{}
    \RightLabel{\axiom}
    \UnaryInfC{$\Gamma \vdash \software$}
    \AxiomC{}
    \RightLabel{\axiom}
    \UnaryInfC{$\Gamma \vdash \hardware$}
    \RightLabel{\andi\rlap{,}}
    \BinaryInfC{$\Gamma \vdash \software \land \hardware$}
    \DisplayProof
  \end{center}
  where $\Gamma = \ens{\software,\hardware}$.
  By adding confidence to the formulas, we derive confidence in
  the assertion that the whole system is correct as a \logic proof:
  \begin{center}
    \AxiomC{}
    \RightLabel{\axiom}
    \UnaryInfC{$\Gamma \vdash \software \colon (0.5,0.2)$}
    \AxiomC{}
    \RightLabel{\axiom}
    \UnaryInfC{$\Gamma \vdash \hardware \colon (0.3,0.01)$}
    \RightLabel{\andi\rlap{.}}
    \BinaryInfC{$\Gamma \vdash \software \land \hardware \colon
      (0.15,0.208)$}
    \DisplayProof
  \end{center}
\end{example}

  If we assume the software system is built from components that are
  difficult to prove reliable (\eg machine learning algorithms),
  but much testing effort has been spent on it, then both true and
  false confidence may be high, as in the example above.
  On the contrary, if there are no reasons not to trust the hardware
  (\eg it is made of very simple, reliable components), but less
  testing effort has been spent on it, then both true and false
  confidence may be lower than those of the software system.

It may seem unnecessary to keep track of false confidence, since our
ultimate goal is to prove system reliability (and not unreliability).
However, there are several cases in which we may want to use it.
For instance, the calculation of how often a volatile system fails
can be translated to a false confidence problem.
Moreover, systems with failsafe mechanisms, \eg \emph{if $A$ works,
use module $B$, else use module $C$}, need to take the unreliability of $A$ into account.
Otherwise, the whole system's reliability could not be correctly expressed,
as it would only depend on $A$ and $B$.
Hence, the optimisation would ignore the reliability of $C$.

\begin{remark}
  In \Cref{fig:more_rules}, negation is an involution, and the reader
  familiar with fuzzy logics will have noticed the product T-norm and
  its dual probabilistic sum T-conorm in the rules $\andi$ and $\ori$.
  Negation is an involution of $[0,1]$ in fuzzy logics, and T-norms
  and T-conorms are the standard interpretations of conjunction and
  disjunction in fuzzy logics, which hints at a deep connection
  between our approach and fuzzy logics.
  However, implication is not interpreted as a residual, which again
  differentiates our approach from fuzzy logics.
\end{remark}

\subsection{Interpretation as Random Variables}
\label{subsec:qcl:prob}

In this section, we justify the \logic proof rules by giving formulas
a probabilistic semantics and showing that these rules are sound.

We start with some measure-theoretic conventions.
We write $\mbool$ to mean the discrete measurable space over the
two-point set $\B=\{\top,\bot\}$.
Boolean algebraic operations over $\mbool$ are
denoted by $\wedge$, $\Rightarrow$, etc.
(There should be no possible confusion with formulas.)
For a probability space $(\smsp,\sfld,\meas)$ (or $\smsp$ for
short), $\Meas(\smsp,\mbool)$ denotes the set of $\mbool$-valued
random variables.
A \emph{context} is a $\Prop$-indexed family of
$\mbool$-valued random variables, given as a function
$\rho \colon \Prop \to \Meas(\smsp,\mbool)$.
\begin{definition}
  Given a space $(\smsp,\sfld,\meas)$,
  we inductively extend a context $\rho$
  to a $\Form$-indexed family of
  $\mbool$-valued random variables
  $\bar{\rho} \colon \Form \to \Meas(\smsp,\mbool)$:
  $$
    \bar{\rho}(A) = \rho(A),
    \enspace \bar{\rho}(\top)(x) = \top,
    \enspace \bar{\rho}(\bot)(x) =  \bot,
    \enspace \bar{\rho}(\phi \imply \psi)(x) =
      \bar{\rho}(\phi)(x) \imply \bar{\rho}(\psi)(x).
  $$
  The \emph{semantics} $\sem{\phi}{\smsp,\rho}$
  of a formula $\phi$ in a space $\Omega$ and context $\rho$ is defined to be
  the probability $\probeq{\bar{\rho}(\phi)}{\top}$
  of $\phi$ being true under $\rho$.
  We say that $\semconf\smsp\rho\phi (t,f)$ holds
  in $\smsp$ and $\rho$
  if $\sem{\phi}{\smsp,\rho} \in [t,1-f]$.
  We say that a sequent $\Gamma\vdash\phi \colon c$ holds in
  $\smsp$ and $\rho$,
  if $\semconf\smsp\rho\phi c$ holds in $\smsp$ and $\rho$
  whenever all $\psi \colon c'$ in $\Gamma$ hold
  in $\smsp$ and $\rho$.
\end{definition}
In other words, the semantics of $\phi$ is the measure of the space on
which $\phi$ holds, and $\phi \colon (t,f)$ holds if $\phi$ is true on
at least $t$ of the space and false on at least $f$.

From here on, we only consider \emph{independent} contexts, \ie
$\rho$'s such that the random variables $\rho(A)$ are mutually
independent for all atomic propositions $A$.

The following lemma lifts independence of $\rho$ to $\bar{\rho}$.
\begin{lemma}
  \label{lem:indep}
  Let $\smsp$ be a space and $\rho$ an independent context.
  If $\phi$ and $\psi$ share no atomic
  propositions, then for all $S,T \subseteq \B$, the following
  holds:
  $$\prob{\bar{\rho}(\phi)\in S\wedge \bar{\rho}(\psi)\in T} =
  \probin{\bar{\rho}(\phi)}{S}\probin{\bar{\rho}(\psi)}{T}.$$
\end{lemma}

\begin{proof}
  By strengthening the proposition to finitely many $\phi$'s and
  $\psi$'s, then by induction on: max depth of $\phi$'s, number of
  $\phi$'s of max depth, max depth of $\psi$'s, and number of $\psi$'s
  of max depth (with lexicographic order).
\end{proof}

We can finally prove soundness of the rules.
\begin{lemma}
  For all rules in \Cref{fig:rules}, formulas $\phi$ and $\psi$
  that share no atomic propositions, spaces $\smsp$, and independent
  contexts $\rho$, if the premise sequents hold in $\smsp$ and $\rho$,
  then so does the conclusion.
\end{lemma}

\begin{proof}
  Simple computations relying on Lemma~\ref{lem:indep}.
\end{proof}

\begin{corollary}
  \label{cor:sound}
  If $\phi$ is linear (each atomic proposition appears at most once)
  and a proof $\pi$ of $\Gamma \vdash \phi \colon c$ only uses base
  rules and introduction rules, then
  $\Gamma \vdash \phi \colon c$ holds
  in all spaces $\smsp$
  and independent contexts $\rho$.
\end{corollary}

\section{Translating System Architectures to Proofs}
\label{sec:trap}
In this section, we translate \acp{ft}~\cite{vesely1981fault} to
\ac{qcl} proof trees.
This allows us to use a system's architecture---modelled as an
\ac{ft}---in our solution to the \ac{trap}.
The way this translation works is close to quantitative fault tree
analysis, where \acp{ft} are equipped with fault probabilities.
In our translation, these fault probabilities are translated to
confidences in \ac{qcl} proofs.

\begin{definition}
  A \emph{fault tree} is a tree whose leaves are called \emph{basic
  events}, and whose nodes, called \emph{gate events}, are either
  \texttt{AND} or \texttt{OR} gates.
\end{definition}
Basic events represent independent components of a system, and the
tree structure represents how faults propagate through the system.
The system fails if faults propagate through the root node.
The usual definition of fault trees is more general than the one we
give here, but we use this one for simplicity.

\begin{figure}
  \begin{subfigure}[c]{0.39\textwidth}
    \centering
    \begin{tikzpicture}[circuit logic US]
      \matrix[
        column sep=0,row sep=7mm
      ]
      {
        & & & \node [rotate=90,or gate] (o) {}; \\
        & \node [rotate=90,and gate] (a1) {}; & & & & \node [rotate=90,and gate] (a2) {}; \\
        \node (n1) {$A$}; & & \node (n2) {$B$}; & & \node (n3) {$C$}; & & \node (n4) {$D$}; \\
      };
      \draw (a1.input 1) -- ++(down:3mm) -| (n1.north);
      \draw (a1.input 2) -- ++(down:3mm) -| (n2.north);
      \draw (a2.input 1) -- ++(down:3mm) -| (n3.north);
      \draw (a2.input 2) -- ++(down:3mm) -| (n4.north);
      \draw (o.input 1) -- ++(down:3mm) -| (a1.output);
      \draw (o.input 2) -- ++(down:3mm) -| (a2.output);
      \draw (o.output) -- ++(up:3mm);
    \end{tikzpicture}
    \caption{~}
    \label{fig:ft}
  \end{subfigure}
  \begin{subfigure}[c]{0.59\textwidth}
    \centering
    \AxiomC{}
    \UnaryInfC{$\Gamma \vdash A$}
    \AxiomC{}
    \UnaryInfC{$\Gamma \vdash B$}
    \BinaryInfC{$\Gamma \vdash A \lor B$}
    \AxiomC{}
    \UnaryInfC{$\Gamma \vdash C$}
    \AxiomC{}
    \UnaryInfC{$\Gamma \vdash D$}
    \BinaryInfC{$\Gamma \vdash C \lor D$}
    \BinaryInfC{$\Gamma \vdash (A \lor B) \land (C \lor D)$}
    \DisplayProof
    \caption{~}
    \label{fig:ft_to_pt}
  \end{subfigure}
  \caption{A fault tree and its translation as a proof tree}
\end{figure}
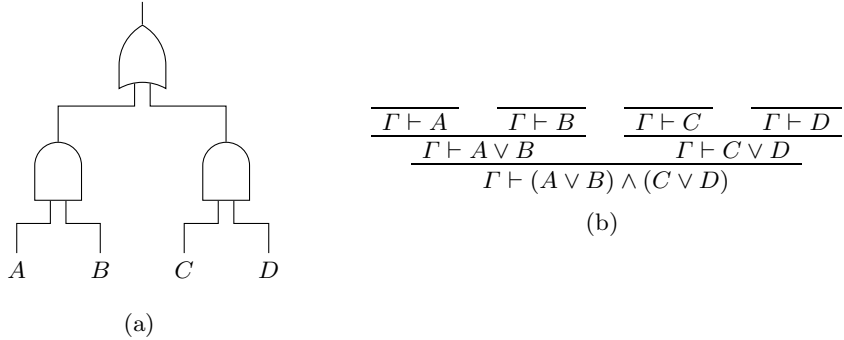

\begin{example}
  \label{ex:ft}
  The \ac{ft} in \Cref{fig:ft} represents a system composed of
  four basic components $A$, $B$, $C$, and $D$.
  For a fault to propagate through the system and become a failure,
  either both $A$ and $B$ have to fail, or both $C$ and $D$ (\eg, $A$
  and $B$ could be redundant components, doubled to increase
  reliability).
\end{example}

In \emph{quantitative fault tree analysis}~\cite{ericson1999fault},
failure probabilities are assigned to basic events, and they propagate
through event gates as if mutually independent.
In other words, if the failure probabilities of an \texttt{AND} gate's
inputs are $a$ and $b$, then its output failure probability is $a b$,
and $a + b - a b$ for an \texttt{OR} gate.

We translate fault trees to \logic proof trees as follows:
The set $\Prop$ of atomic propositions collects all the names of basic events.
$\Gamma$ consists of $A\colon(t_A,f_A)$ for each $A \in \Prop$.
\texttt{AND} gates are
translated to $\ori$ rules, and \texttt{OR} gates to $\andi$ rules.

The reason for this dualisation is straightforward: while a fault tree
represents how faults propagate, proof trees represent confidence in a
system's reliability, \ie, how \emph{absence of faults} propagates:
the true confidence in each atomic proposition in $\Gamma$ now represents
reliability of the component, and the true confidence in the
conclusion represents reliability of the whole system.

\begin{example}
  The translation of the fault tree of \Cref{fig:ft} is shown in
  \Cref{fig:ft_to_pt} (with confidences left out for readability).
  Here,
  $\Prop = \ens{A,B,C,D}$
  and the true confidence in the conclusion of the proof is
  $t_A t_B + t_C t_D - t_A t_B t_C t_D$.
  We can thus link increases in the reliability of components to
  increases in reliability of the whole system.
\end{example}

Note that the translation of a fault tree only uses base rules and
introduction rules (more precisely, only $\axiom$, $\andi$, and
$\ori$).
This is partly because we only consider \texttt{AND} and \texttt{OR} gates,
but more essentially, basic events of fault trees are considered atomic and thus there is no need to eliminate logical connectives.
Moreover, an assignment of failure probabilities to basic events
translates to a context $\rho$ in \ac{qcl} terms.
Since all basic events are considered independent in an \ac{ft}, their
translation gives an independent context $\rho$.
Therefore, the translation of a fault tree always verifies
Corollary~\ref{cor:sound}, and our interpretation as a \ac{qcl} proof
tree is sound for any assignment of failure probabilities to basic
events.
This means that, if the confidence in all basic components corresponds
to their reliability, then the confidence of the whole proof cannot
overshoot the whole system's reliability.

Since we translate fault probabilities to confidences, and fault
probabilities are directly linked to reliability, we may use
``reliability'' and ``confidence'' interchangeably in the following,
\eg when we feel that reliability conveys a better intuition than
confidence.

\section{Solving the Test Resource Allocation Problem}
\label{sec:optim}
In this section, we show how to optimise confidence in the conclusion
of a \ac{qcl} proof.
This gives a solution to the \ac{trap} through the
translation of \acp{ft} to \ac{qcl} proofs that was described in
\Cref{sec:trap}.

In order for this approach to be usable in practice, the user has to
be able to specify two input parameters: the \ac{ft} that
represents the system's architecture, and functions describing how confidence in each
component's reliability grows by spending resource on it.
\acp{ft} are commonly used in the industry, so modelling a system
using them should be no problem in practice.
We first describe the latter input parameter in
\Cref{subsec:optim:funcs} and then give our solution to the \ac{trap}
in \Cref{subsec:optim:trap}.

\subsection{Confidence Functions}
\label{subsec:optim:funcs}

Increasing confidence in a proof's conclusion
requires an increase in its premises' confidences.
The cost of increasing confidence may vary among premises;
when thinking in terms of systems (rather than proofs),
for instance, increasing trust in a machine learning algorithm may require more effort than improving
hardware reliability.
This is a well-known problem, for which many solutions have been
designed, especially \acp{srgm}, which are based on mathematical
modelling of faults~\cite{yamada1985software}.
Here, however, we do not choose a particular fault model and instead
introduce the following abstract notion, which makes the approach versatile.

\begin{definition}
  A \emph{confidence function} is a non-decreasing function
  $f : \R_+ \to \C$ (equipped with $\sqsubseteq$).
\end{definition}

The equality $f(r) = c$ means that after spending $r$ resources on a
formula, one will have confidence $c$ in the formula.
The monotonicity condition above enforces that, by spending more
resources, confidence should not decrease.

Note that $f(0)$ can be different from $(0,0)$, which corresponds to
the fact that engineers usually have some confidence in the components
they use.
This feature also makes it easy to use our approach in continuous
development, by using confidence functions $f_s(r) = f(r+s)$ where $s$
is the amount of resource that has been already spent to test a
component.

Note also that a confidence function may increase the false
confidence, theoretically capturing the fact that faults may found
by testing.
For an application on the \ac{trap}, however, we assume that faults
will be fixed and thus false confidence always stays at $0$.
In the following, we thus define confidence functions as increasing
functions $f \colon \R_+ \to [0,1]$, which represents the true
confidence, and assume the false confidence is always $0$.

\paragraph{Designing confidence functions}
Of course, expert knowledge on a component can be used to give a good
estimate of confidence functions, but other techniques, such as
\emph{defect prediction}~\cite{kamei2016defect}, exist for when
knowledge is limited.

When testing is the canonical way to increase confidence,
notions of test coverage serve as a good estimate of confidence.
If we have a hardware test suite of $n$ tests that achieves 100\% coverage
(but not enough budget to execute them all), and each test costs $r_0$ resources,
then the coverage achieved by spending $r$ resources in testing can be
estimated as the confidence $f(r) = \min(r / n r_0, 1)$.

If we do not have such a test suite, then a reasonable way to model confidence is to assume uniform random testing.
There we assume that each test covers a randomly sampled fraction $p$
of the input space, but parts of it might be already covered by
previous tests.
If running a test costs $r_0$ resources, then a good estimate of confidence function is $f(r) = 1 - (1-p)^{r/r_0}$.

If more is known about the component, then it is possible to design
confidence functions that are better suited for this component.
In particular, if we have some \emph{a priori} knowledge about fault
distributions, then it is possible to use
\acp{srgm}~\cite{yamada1985software} as confidence functions.

\subsection{The Optimisation Problem}
\label{subsec:optim:trap}

We now formulate the \ac{trap} as an optimisation problem in terms of
\logic as follows: given a \logic proof, a confidence function for
each premise,
and a resource budget to spend, how should we spend the budget on the
different premises to maximise confidence in the proof's conclusion?

We only consider the problem of optimising true confidence because
the application we are aiming at is about reliability.
However, with the same ingredients, we could define similar
optimisation problems.
For example, we could try to optimise total confidence $t+f$ under
limited resources, or try to minimise resources spent to reach
a given confidence objective (either in true or total confidence).

We begin with a simple observation: if $\phi_1 \colon c_1, \ldots,
\phi_n \colon c_n \vdash \phi \colon c$ is provable, then $c$ is a
non-decreasing function of the $c_i$'s (for the confidence order
$\sqsubseteq$).
Hence, increases in the $c_i$'s confidence lead to increases in $c$.

Because, for the translation of an \ac{ft}, the true confidence of
the conclusion has to be a function $f(t_1,\ldots,t_n)$ of the
true confidences of the hypotheses, if the confidence of each
hypothesis is given by applying a confidence function $f_i$ to an
amount of resources $r_i$ spent on that hypothesis, then the true
confidence of the conclusion is itself a function $f(f_1(r_1),\ldots,
f_n(r_n))$ of the amount of resources spent on the hypotheses.

The problem is thus the following: given an initial condition $r_1,
\ldots,r_n$, confidence functions $f_1,\ldots,f_n$, a proof of
$\enscomp{\phi_i \colon (c_i,0)}{i \in n} \vdash \phi \colon
(f(c_1,\ldots,c_n),-)$, and a budget $r$, maximise $f(f_1(r_1+r'_1),
\ldots, f_n(r_n+r'_n))$ under $r'_i \geq 0$ for all $i \in n$ and
$\sum_{i \in n} r'_1 \leq r$.

We thus reduce the \ac{trap} to a classic constrained
optimisation problem, which we can solve using well-known algorithms.
In our implementation, we use simulated
annealing~\cite{van1987simulated}, but any other method (such as
CMA-ES~\cite{hansen2001completely} or Lagrange
multipliers~\cite{bertsekas2014constrained}) would work too.

\begin{example}
  \label{ex:trap}
  Take the \ac{qcl} proof from \Cref{ex:ft}.
  Suppose that the confidence functions of components follow
  $f(r) = 1 - 1/2^r$,
  the amount of resources already spent on the components
  are $0$ for $A$,
  $5$ for $B$ and $C$, and $10$ for $D$, and we have a test resource
  budget of $10$. Then we want to maximise
  \begin{align*}
    & ( 1 - 1 / 2^a )     ( 1 - 1 / 2^{5+b} ) +
      ( 1 - 1 / 2^{5+c} ) ( 1 - 1 / 2^{10+d} ) \\
    & \quad - ( 1 - 1 / 2^a )     ( 1 - 1 / 2^{5+b} )
      ( 1 - 1 / 2^{5+c} ) ( 1 - 1 / 2^{10+d} )
  \end{align*}
  under the constraints $a,b,c,d \geq 0$, and $a+b+c+d \leq 10$.
  There are two major points to note here.
  First, due to the fact that $\ori$ requires both disjuncts to be
  proved, the optimisation will try to increase confidence of
  \emph{both} $A$ and $B$, rather than choose one.
  Second, since we take system structure into account, the algorithm
  can give $B$ and $C$ different budgets, even though they share the
  same initial confidence and confidence function.
\end{example}

Our approach has significant advantages over other \ac{trap}
solutions.
First, it makes use of the system's architecture, which is not the
case of most approaches.
Even other approaches that take system architecture into
account generally only consider simple architectures, such as
\emph{parallel-series architecture}~\cite{zhang2017constraint}.
These architectures can be directly translated to \acp{ft}, but the
converse is not possible without duplicating modules, which puts
artificial weight to these duplicated modules.
Moreover, we explained how to convert an \ac{ft} to a \ac{qcl} proof,
but our algorithm is not limited to \acp{ft} and would work on other
proofs.

Another advantage of our method is that it is not tied to any specific
confidence function.
The main advantage of this generality is that it allows the user to
pick different confidence functions for different components.
In particular, this approach should be helpful when allocating test
resources for \acp{cps}, where some components are software, while
others are hardware, which most likely require to be modelled using
different confidence functions.

\section{Experimental Results}
\label{sec:exp}
In this section, we describe the results of our experiments, showing
that our tool \Astrahl\footnote{The code and experimental data are
publicly available on
\url{https://github.com/ERATOMMSD/qcl_tap_2021}.} can increase system
reliability more consistently than others.
To demonstrate the tool's performance, we designed two experiments.
The first one compares \Astrahl's confidence gain to other \ac{tra} strategies.
The second, more involved experiment tests
whether the increase in confidence provided by \Astrahl is linked to
an increase in system reliability.
Given an \ac{ft} and confidence functions, we simulate existence of component faults,
before splitting a fixed testing budget according to different \ac{tra} strategies (one of which is our
approach).
We then mimic component testing according to the allocated
budget and fix faults if they are found, thereby increasing system reliability.
Our evaluation repeats the probabilistic process to test
which method gives the best reliability on average.

We developed \Astrahl, which implements the \ac{tra} algorithm
described in Section~\ref{sec:optim}.
It takes as input JSON descriptions of the fault tree and the
confidence functions (as parse trees), an initial condition (a float
for each basic event), and a budget (a float), and returns a splitting
of the budget between the different basic events (a float for each
basic event).

This section evaluates our claims and analyses \Astrahl's system confidence gains to other, more naive approaches.
We first ask how much confidence we can gain by using \Astrahl, rather
than simpler \ac{tra} approaches.
Then, we test whether using \Astrahl can increase system reliability
in practice.
Specifically, this section will investigate the following two research questions:
\begin{enumerate}
    \item[RQ1] Given a certain \ac{tra} budget, how much is the calculated confidence gain when using \Astrahl and how do these figures compare to alternative \ac{tra} methods?
    \item[RQ2] Does \Astrahl's gain in confidence translate to a gain
      in system reliability in a practical scenario where testing
      practice is simulated?
\end{enumerate}

\paragraph{Alternative \ac{tra} approaches}
There exist numerous solutions to test resource splitting, however
some of the most common ones are the uniform and proportional resource
allocation strategies (see \Cref{fig:naivetrap}), as they do not
require knowledge of the system structure or fault distribution.
\emph{Uniform} \ac{tra}, for instance, evenly distributes the
available resources among the candidate components.
This technique is completely agnostic of the current system and component confidences.
\emph{Proportional} \ac{tra} on the other hand aims to take current
component confidence into account and provide proportionally more
resources to components in which we have lower confidence.
Although it uses current confidences for resource allocation, the
system's structure is still not considered.

\newcommand{\UniformColor}{NavyBlue!75}
\newcommand{\PropColor}{YellowOrange}
\newcommand{\IterPropColor}{OliveGreen}

\begin{figure}
  \centering
    \begin{tikzpicture}[scale=0.01, thick,-latex]

        \pgfmathsetmacro{\yscale}{300}
        \pgfmathsetmacro{\offset}{5}

        \draw[] (-10, 0) node[left, black] {0.0} -- (575, 0) node[right] {\scriptsize resources};
        \draw[] (0, -10)  node[below, black] {0} -- (0, {\yscale + 10}) node[above] {\scriptsize confidence};
        \draw[-,gray, dashed, ultra thin] (0, \yscale) node[left, black] {1.0} -- (575, \yscale);

        \draw[-,semithick,domain=0:575,smooth,variable=\x,black] plot ({\x}, {(1 - 0.99 ^ (\x + 1)) * \yscale});
        \node[] at (500, \yscale+15) {\scriptsize $c(r) = 1 - 0.99^{(r + 1)}$};

        \foreach \xstep in {100,200,...,500}
        {
            \pgfmathsetmacro{\yvalue}{(1 - 0.99 ^ (\xstep + 1))}
            \draw[-,gray,dashed, ultra thin] (\xstep, {\yvalue * \yscale}) -- (\xstep, -5) node[below] {\tiny \xstep};
        }

        \foreach \xstep/\ylabel in {100/0.64,200/0.87,300/0.95}
        {
            \pgfmathsetmacro{\yvalue}{(1 - 0.99 ^ (\xstep + 1))}
            \draw[-,gray,dashed, ultra thin] (\xstep, {\yvalue * \yscale}) -- (\xstep, -5) node[below] {\tiny \xstep};
            \draw[-,gray,dashed, ultra thin] (0, {\yvalue * \yscale}) node[left] {\tiny \ylabel} -- ++(\xstep, 0);
        }

        \foreach \xstep in {50,150,250}
        {
            \draw[\UniformColor] (\xstep, {(1 - 0.99 ^ (\xstep + 1)) * \yscale + \offset}) -- ++(100,0);
        }

        \pgfmathsetmacro{\cf}{(1 - 0.99 ^ (50 + 1)) };
        \pgfmathsetmacro{\chf}{(1 - 0.99 ^ (150 + 1)) };
        \pgfmathsetmacro{\cthf}{(1 - 0.99 ^ (250 + 1)) };
        \pgfmathsetmacro{\vf}{300 * (1 - \cf) / ((1 - \cf) + (1- \chf) + (1- \cthf))};
        \pgfmathsetmacro{\vhf}{300 * (1 - \chf) / ((1 - \cf) + (1- \chf) + (1- \cthf))};
        \pgfmathsetmacro{\vthf}{300 * (1 - \cthf) / ((1 - \cf) + (1- \chf) + (1- \cthf))};
        \draw[\PropColor, thick]  (050, {\yscale *     \cf }) -- ++(\vf,0);
        \draw[\PropColor, thick]  (150, {\yscale *   \chf }) -- ++(\vhf,0);
        \draw[\PropColor, thick]  (250, {\yscale * \cthf }) -- ++(\vthf,0);

        \foreach \xstep in {50,150,250}
        { \node at (\xstep, {(1 - 0.99 ^ (\xstep + 1)) * \yscale}) {\textbullet};
        }

        \matrix [draw=black!50,thin,fill=white,above left,inner sep=0.2em] at (550, 5) {
            \draw[\UniformColor,thick] (0,0) -- ++(.5,0) node[black,right] {\tiny Uniform}; \\
            \draw[\PropColor,thick] (0,0) -- ++(.5,0) node[black,right] {\tiny Proportional}; \\
        };

    \end{tikzpicture}
    \caption{Allocation of test budget according to common strategies}
    \label{fig:naivetrap}
\end{figure}
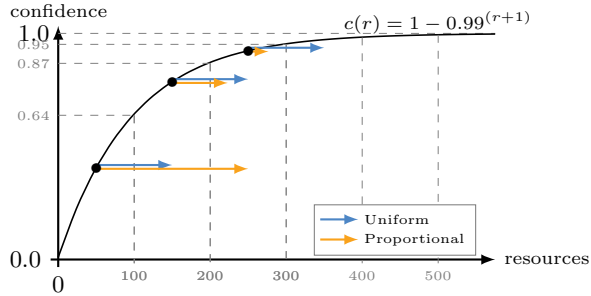

\subsection{RQ1: Theoretical Evaluation}
Naive approaches might coincidentally be equally good as elaborate
techniques, given the right system architecture and initial
confidences.
We therefore chose to perform our comparison on a set of randomly generated initial confidences (later referred to as \acp{sp}) and \acfp{ft}.
It is natural to expect \Astrahl's insight into system structure and
component confidences to outperform naive strategies as the systems
grow in size and complexity.\footnote{Functional optimisation may not
be as efficient in larger dimensions, but even a naive estimate should
give a better result than completely ignoring system structure.}
Therefore, we only verify \Astrahl's superiority on relatively
small systems and simple confidence functions.
We thus fixed an \ac{ft} size of six devices connected by five binary
\texttt{AND} and \texttt{OR} gates.
Furthermore, confidences behave according to the function $c(r) = 1 -
0.99^{(r + 1)}$ for all devices, where $r$ represents the
invested resources and $c$ the confidence, as displayed in
\Cref{fig:naivetrap}.

Using these settings, we generated 200 \acp{ft} and instantiated each
with 100 random \acp{sp} in the range of 100 to 300, corresponding to
initial confidences between approximately $0.64$ and $0.95$.
Using this data set we let \Astrahl and its competitors distribute
total budgets of size $1$, $10$, $50$, $100$, $250$, $500$, and
$1000$.

\subsection{RQ2: Empirical Evaluation}

To address RQ2 it is necessary to create an evaluation setting that
allows the simulated distribution of (hidden) component faults, their
(potential) discovery through testing or experimentation effort and
subsequent removal, and finally a calculation of the system confidence
based on the remaining, undiscovered faults.
Our approach is based on the probabilistic creation of \acp{fd}, \ie
assignments of faults to components according to their respective
confidences.
These faults will be probabilistically found and removed by allocating resources to a component, simulating \eg experimentation or testing.
Our hypothesis is that, given initial component confidences that
reflect the components' reliabilities, \Astrahl should be able to
outperform its competitor algorithms and on average
lead to higher overall system confidence.

The evaluation process is split into three phases.
First, faults are assigned to components according to geometric
distributions with parameter $p = 1 - c$, where $c$ is our initial
confidence in the component.
Therefore, components in which we have more confidence will on average
contain fewer faults.
Next, the faults are removed probabilistically during a ``testing
phase'' as follows.
We arbitrarily assume that each test costs 10 resources\footnote{It
would equally be possible to assume a test costs one resource and
scale the budget.}.
Each fault has an observability of $0.1$, \ie each test
has a $10\%$ chance to detect this particular fault.
When a \ac{tra} strategy assigns $r$ resources to a component with $n$ faults, $t = \lfloor\frac{r}{10}\rfloor$ full tests are run on it.
Each test has a $10\%$ chance to find and remove each of a component's $n$ faults.
If $r > 10 t$, \ie there is remaining budget, a ``partial'' test is
run with proportionally reduced chance to find faults.
After this phase, we end up with $n'$ faults in each component.
Finally, the system fault probability is calculated.
As above, during operation each fault's observability is $0.1$, so a
component's failure probability can be calculated as $1 - 0.9^{n'}$
if it contains $n'$ faults.
The entire system's failure probability can then be calculated using
all components' failure probabilities and the standard propagation of
fault probabilities in \acp{ft}.

Due to the probabilistic nature of this evaluation we repeated this
process for 50 \acp{ft}, 50 \acp{sp} (range 10 to 70)\footnote{We
used smaller confidence so that components will usually contain
faults.} for each \ac{ft} and 50 random \acp{fd} for each \ac{sp},
totalling to $125{,}000$ \acp{fd}.
We computed test resource allocations using test budgets of $60$,
$120$, $240$, $360$, $480$ and $600$, executed the testing and fault
removal process $100$ times for each \ac{fd} and test budget, and
subsequently calculated the average \ac{ft} failure probability, for a
total of $75{,}000{,}000$ computations.

\subsection{Evaluation Results}

The experiment results for RQ1 are shown in \Cref{tab:theoreticalExperiment} and \Cref{fig:theoreticalexperiment}.
\Cref{tab:theoreticalExperiment} shows the average system reliability
according to the total budget and relative difference to \Astrahl's
score $\frac{(1-r)-(1-r')}{1-r} = \frac{r'-r}{1 - r}$, where $r$ is
the reliability computed by \Astrahl, and $r'$ that computed by its
competitor (this measure is closer to intuition than $(r-r')/r$ when
both confidences are close $1$).
The error bars in \Cref{fig:theoreticalexperiment} represent mean
squared error in system reliability.
Note that \Astrahl outperforms each of the competitors independent of the budget size.
It is also noteworthy that although with higher budgets the system becomes very reliable independent of the strategy, the relative performance increase of \Astrahl when compared to its competitors grows significantly.
In other words, spending a large amount of resources increases
obviously the system performance, but it is still better to follow
\Astrahl's suggestions.

\begin{figure}
    \begin{subfigure}[c]{0.46\textwidth}
      \centering
         { \scriptsize
            \begin{tabular}{r @{}c c @{}c c r @{}c c r}
                \toprule
                & \phantom{i} &   Astrahl &\phantom{i} & \multicolumn{2}{c}{Uniform} & \phantom{i} & \multicolumn{2}{c}{Proportional}\\
                \cmidrule{3-3}\cmidrule{5-6}\cmidrule{8-9}
                Budget & &Score&&Score&Diff \%&&Score&Diff \%\\
\midrule
                1~~    &&   .8445 &&   .8442 &           -0.19 &&  .8442 &        -0.19 \\
                10~~   &&   .8498 &&   .8465 &           -2.20 &&  .8471 &        -1.80 \\
                50~~   &&   .8697 &&   .8565 &          -10.13 &&  .8593 &        -7.98 \\
                100~~  &&   .8884 &&   .8682 &          -18.10 &&  .8729 &       -13.89 \\
                250~~  &&   .9226 &&   .8976 &          -32.30 &&  .9053 &       -22.35 \\
                500~~  &&   .9544 &&   .9329 &          -47.15 &&  .9400 &       -31.58 \\
                1000~~ &&   .9812 &&   .9711 &          -53.72 &&  .9730 &       -43.62 \\
                \bottomrule
            \end{tabular}
        }
    \caption{~}
    \label{tab:theoreticalExperiment}
    \end{subfigure}\hfill
    \begin{subfigure}[c]{0.53\textwidth}
      \centering
    \scalebox{.52}{
      \input{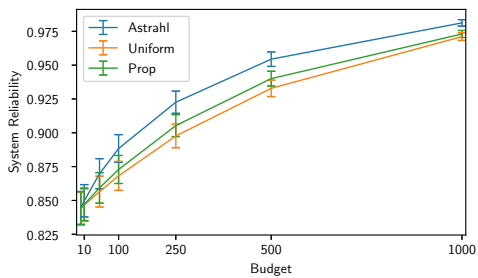}
    }
    \caption{~}
    \label{fig:theoreticalexperiment}
    \end{subfigure}
  \caption{Theoretical evaluation: average system reliability and relative difference}
\end{figure}

\Cref{tab:empiricalExperiment} and \Cref{fig:empiricalexperiment} display the results of the empirical evaluation (RQ2).
Here, the error bars correspond to mean squared error of the average
over all \acp{fd} (for each \ac{sp}).
As can be seen, also here \Astrahl outperforms other \ac{tra} strategies.
Interestingly though, \Astrahl's relative advantage is not as high.
An initial investigation suggests the cause for this observation at the discrete nature of the evaluation setting,
where in many cases all faults of a component are removed, which leads to full confidence in this component.

\begin{figure}
    \begin{subfigure}[c]{0.46\textwidth}
      \centering
        { \scriptsize
            \begin{tabular}{r @{}c c @{}c r r @{}c r r}
                \toprule
                & \phantom{i} &   Astrahl &\phantom{i} & \multicolumn{2}{c}{Uniform} & \phantom{i} & \multicolumn{2}{c}{Proportional}\\
                \cmidrule{3-3}\cmidrule{5-6}\cmidrule{8-9}
                Budget & &Score&&Score&Diff \%&&Score&Diff \%\\
\midrule
               60~~  &&   .8982 &&   .8890 &           -9.04 &&  .8887 &        -9.33 \\
               120~~ &&   .9146 &&   .9000 &          -17.10 &&  .8995 &       -17.68 \\
               240~~ &&   .9380 &&   .9188 &          -30.97 &&  .9179 &       -32.42 \\
               360~~ &&   .9541 &&   .9341 &          -43.57 &&  .9329 &       -46.19 \\
               480~~ &&   .9657 &&   .9466 &          -55.69 &&  .9451 &       -60.06 \\
               600~~ &&   .9743 &&   .9567 &          -68.48 &&  .9550 &       -75.10 \\
                \bottomrule
            \end{tabular}
        }
        \caption{~}
        \label{tab:empiricalExperiment}
    \end{subfigure}\hfill
    \begin{subfigure}[c]{0.51\textwidth}
      \centering
        \scalebox{.52}{
            \input{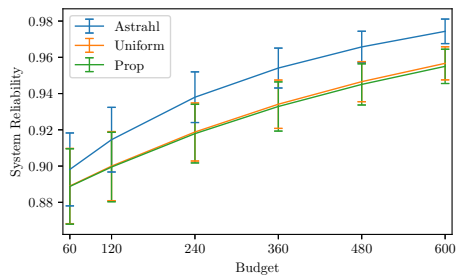}
        }
        \caption{~}
        \label{fig:empiricalexperiment}
    \end{subfigure}
    \caption{Empirical evaluation: system reliability and relative difference}
\end{figure}

Summarising our evaluations it can be said that \Astrahl is better-suited for identifying where to place test effort than alternative approaches.
The experiments show significant relative gains in both theoretical and practical approaches, even for rather small, straightforward systems as in our setting.
For more complex systems, we expect \Astrahl's insight into the
system's structure and the components' confidence should make its
advantage even clearer, although this has yet to be validated by
experimental results.

\section{Conclusion and Future Work}
We have defined \emph{\acl{qcl}}, which represents confidence in
assertions and have argued that this logic can help us take system architecture
into account when solving the \acf{trap} and shown the validity of the
approach through experimental results.
We have also argued that this approach is widely applicable, \eg,
because it does not rely on particular assumptions about fault
distributions.

The simplicity and versatility of our approach makes it possible to
tackle different problems with the same ingredients.
An obvious possible future work is to study the \ac{trap} in different
settings, for example by implementing multi-objective optimisation, or
by studying it in a broader setting, where the confidence gained by
running a test depends on the result of the test.
We should also experimentally validate our expectation on the scalability of our approach in industry-scale case studies.
It would also be interesting to see how solving the \ac{trap} when
optimising the total confidence $t+f$ compares to solving it with the
current setting, especially on volatile systems.
Another possible direction is to study how this approach can be used
to solve test prioritisation between different components of a system.

We also want to investigate the logic itself more thoroughly from a
purely logical point of view.
For example, by changing the interpretation of connectives in
three-valued logic, or using different T-norms and T-conorms in the
definitions of the rules.
Another interesting aspect would be to investigate its links with
fuzzy logics and Dempster-Shafer theory deeper, as there seems to be
some deep connections.
In particular, ties to fuzzy logics would give a bridge between a
logic about confidence and a logic about truth, which could help us
develop \ac{qcl} further.

\bibliographystyle{splncs04}
\bibliography{main}

\end{document}